%% file: Main.tex
\documentclass[11pt]{article}

\usepackage{geometry}
\usepackage{indentfirst}
\usepackage{setspace}
\usepackage{amsmath,amsfonts,amssymb,amsthm,epsfig,epstopdf,titling,url,array,mathrsfs,tikz,tikz-cd}
\usepackage{tkz-graph}
\usepackage{color}

\usetikzlibrary{cd}

\usepackage{hyperref}
\usepackage{verbatim}

\usetikzlibrary{arrows,calc,patterns,positioning,shapes}
\usetikzlibrary{decorations.pathmorphing}

\newcommand{\bbc}{\mathbb{C}}
\newcommand{\bbn}{\mathbb{N}}

\newcommand{\bbz}{\mathbb{Z}}
\newcommand{\bbh}{\mathbb{H}}

\newcommand{\cala}{\mathcal{A}}

\newcommand{\pd}[1]{\frac{\partial}{\partial #1}}
\newcommand{\pds}[2]{\left(\pd{#1}\right)^{#2}}
\newcommand{\Y}{\mathbf{Y}}
\newcommand{\A}{\mathbf{A}}
\newcommand{\kuohao}[1]{\left(#1\right)}

\newcommand{\vol}[1]{\frac{- dz_{#1} \wedge d\bar{z}_{#1}}{2i \op{Im}\tau}}

\newcommand{\op}{\operatorname}

\theoremstyle{plain}
\newtheorem{thm}{Theorem}[section]
\newtheorem{thm-defn}{Theorem/Definition}[section]
\newtheorem{lem}[thm]{Lemma}
\newtheorem{lem-defn}[thm]{Lemma/Definition}
\newtheorem{prop}[thm]{Proposition}
\newtheorem{cor}[thm]{Corollary}

\theoremstyle{definition}
\newtheorem{defn}[thm]{Definition}

\newtheorem{eg}[thm]{Example}

\theoremstyle{remark}
\newtheorem{rmk}[thm]{Remark}

\def\Xint#1{\mathchoice
 {\XXint\displaystyle\textstyle{#1}}%
 {\XXint\textstyle\scriptstyle{#1}}%
 {\XXint\scriptstyle\scriptscriptstyle{#1}}%
 {\XXint\scriptscriptstyle\scriptscriptstyle{#1}}%
 \!\int}
\def\XXint#1#2#3{{\setbox0=\hbox{$#1{#2#3}{\int}$}
  \vcenter{\hbox{$#2#3$}}\kern-.5\wd0}}

\tikzset{
       modal/.style={>=stealth',shorten >=1pt,shorten <=1pt,auto,
                     node distance=1.5cm,semithick},
       world/.style={circle,draw,minimum size=1cm,fill=gray!15},
       point/.style={circle,draw,fill=black,inner sep=0.5mm},
       reflexive/.style={->,in=120,out=60,loop,looseness=#1},
       reflexive/.default={5},
       reflexive point/.style={->,in=135,out=45,loop,looseness=#1},
       reflexive point/.default={25},
       }

\begin{document}
		
	\title{\mbox{\bf On the Feynman Graph Integrals on Elliptic Curves} }
	\author{Xiaoxiao Yang}
	\date{}
	\maketitle
	
	\begin{abstract}
		We introduce the regularized integrals for decorated graphs on elliptic curves, which produces an almost holomorphic function on upper half plane. Then we give the graph version of holomorphic anomaly equation to study the anti-holomorphic dependence of the graph integral. Finally, we give a general method to compute the graph integral.
	\end{abstract}
	
	\input{Introduction.tex}

	\input{Intro_of_RI.tex}

	\input{Regularized_Integral_on_Graphs.tex}

	\input{Regularized_Integral.tex}

	\bibliographystyle{alpha}
	\bibliography{Reference.bib}
	
	\noindent{\small Department of Mathematical Science, Tsinghua University, Beijing 100084, P. R. China}
	
	\noindent{\small Email: \tt yangxx19mails.tsinghua.edu.cn}

\end{document}

%% file: Introduction.tex
\section{Introduction}

We introduce the regularized integral for decorated graphs on elliptic curves (see Definition \ref{Definition of R.I})
\begin{equation} \label{Formula of R.I}
	W(\Gamma,n) = \Xint-_{E_{\tau}^m} \Phi_{(\Gamma,n)} \op{vol} \cdot \prod_{e\in L(\Gamma)}W_{n_e},
\end{equation}
 which produces an almost holomorphic\cite{kaneko1995generalized} function on upper half plane $\bbh$. In fact, if a function $\Phi(z;\tau)$ on $\bbc^m \times \bbh$ is elliptic, modular of weight $k$, and smooth away form diagonals, with possibly holomorphic poles on diagonals, then $f(\tau) := \Xint-_{E_{\tau}^m} \Phi(z;\tau) \op{vol} \in \bbc[E_4,E_6][\widehat{E}_2]$ is almost holomorphic and modular of weight $k$ on the upper half plane $\bbh$ \cite{li2021regularized}. 

The anti-holomorphic dependence of $f$ is of polynomial in $\Y = -\frac{\pi}{\op{Im}\tau}$, which satisfies the holomorphic anomaly equation \cite{li2023regularized}
\begin{equation} \label{Holomorphic anomaly}
	\partial_{\Y} \Xint-_{E_{\tau}^n} \Psi = \Xint-_{E_{\tau}^n} \partial_{\Y}\Psi - \sum_{a<b} \Xint-_{E_{\tau}^{n-1}}\op{Res}_{z_a=z_b}\kuohao{(z_a-z_b)\Psi dz_a}.
\end{equation}

Let $G$ be the set of decorated graphs, $V(G)$ be the $\bbc$-linear space spanned by $G$, $W \colon V(G) \to \bbc[E_4,E_6][\widehat{E}_2]$ be the linear map induced by \eqref{Formula of R.I}. Let 
\[\delta(\Gamma,n) = \sum_{e \in E \cup L \atop n_e=0}(\Gamma \setminus e,n) - \frac{1}{2}\sum_{v,w \in V \atop v \neq w} \op{Res}_{w}^{(v)}[1](\Gamma,n)\]
be a differential operator on $V(G)$. Here $\Gamma \setminus e$ is the graph by deleting $e$, $\op{Res}_{w}^{(v)}[k](\Gamma,n)$ is the $k$-shifted residual graph with one vertex $v$ collapsing to another vertex $w$ in $\Gamma$ (see Definition \ref{Definition of k-shifted residual graph}). In this paper, we discuss the graph version of holomorphic anomaly equation:

\begin{thm}[= Theorem \ref{Main Theorem}]
	The following diagram commutes.
	\[\begin{tikzcd}
	V(G) \arrow[r,"W"] \arrow[d,"\delta"] &  \bbc[E_4,E_6] \arrow[d,"\partial_{\Y}"] [\widehat{E}_2]\\
	V(G) \arrow[r,"W"] & \bbc[E_4,E_6][\widehat{E}_2].
\end{tikzcd}\]
\end{thm}

In section 4, we use the idea of iterated integral to give a general algorithm to compute the regularized integral on graphs:
\begin{thm}[=Theorem \ref{Main Thm 2}] \label{Thm2 in intro}
	Let $(\Gamma,n)$ be a decorated graph, and let $v$ be a vertex in $\Gamma$. Then
	\[W(\Gamma,n) = \sum_{w \in N(v)} W(\op{Res}_{w}^{(v)}[0](\bar\delta^{-1}(\Gamma,n))).\]
\end{thm} 

Here $\op{Res}_{w}^{(v)}[0](\bar\delta^{-1}(\Gamma,n))$ is a graph with $|\Gamma|-1$ vertices, $N(v)$ is the neighborhood of $v$ in $\Gamma$ (See Lemma \ref{0-shifted residual graph lemma}). And we can use Theorem \ref{Thm2 in intro} iteratively to compute the graph integral. We will also compute some examples in section 4.

The regularized integrals is closely related to the so-called ordered $A$-cycle integrals studied in the literature \cite{oberdieck2018holomorphic,dijkgraaf1995mirror}. In fact, the regularized integral is exactly the modular completion (in the sense of Kaneko–Zagier \cite{kaneko1995generalized}) of the average over orderings of the ordered $A$-cycle integrals \cite{li2021regularized} .

\paragraph{Acknowledgements} I greatly appreciate Si Li for discussions on this work.

\nocite{dijkgraaf1995mirror}

\nocite{martin1966complex}

\nocite{miyake2006modular}

\nocite{kaneko1995generalized}

\nocite{oberdieck2018holomorphic}

\nocite{mnev2017lectures}

%% file: Intro_of_RI.tex
\section{Regularized Integral on Elliptic Curves}
In this section, we will briefly recall the definition and some important properties of regularized integral. More details about regularized integral or proof of properties mentioned in this section can be found in \cite{li2021regularized} and \cite{li2023regularized}. 

\paragraph{Regularized Integrals}
	Let $\Sigma$ be a compact Riemann surface, possibly with boundary $\partial \Sigma$. 
	And let $D$ be a finite subset of $\Sigma$ which does not meet the boundary.
	Let $\cala^{p,q}(\Sigma,*D)$ be the set of $(p,q)$-forms on $\Sigma$ which are smooth on $\Sigma \setminus D$ but may admit \textit{holomorphic poles} of arbitrary order along $D$. 
	More precisely, for $\omega \in \cala^{p,q}(\Sigma,*D)$ and $p \in D$, let $z$ be a local holomorphic coordinate around $p$ with $z(p) = 0$, then there exists an integer $n \geq 0$ and a smooth $(p,q)$-form $\alpha$ around $0 \in \bbc$ such that $\omega = \alpha /z^{n}$ in a neighborhood of $0$. 
	Let $\cala^{1,q}(\Sigma, \log D) \subset \cala^{1,q}(\Sigma,*D)$ be the set of \emph{logarithm $(1,q)$-forms} along $D$ such that the elements in $\cala^{1,q}(\Sigma,\log D)$ have holomorphic poles with order $ \leq 1$.	
	\begin{lem}[= Lemma 2.1 in \cite{li2021regularized}]\label{Regularized Integral}
		Let $\omega \in \cala^{1,q}(\Sigma,*D)$, there exists $\alpha \in \cala^{1,q}(\Sigma,\log D)$ and $\beta \in \cala^{0,q}(\Sigma,*D)$, such that
		\[\omega = \alpha + \partial \beta.\]
		Moreover, for $q = 1$, if $\omega = \alpha' + \partial \beta'$ is another decomposition of $\omega$, then 
		\[\int_{\Sigma} \alpha + \int_{\partial \Sigma} \beta = \int_{\Sigma} \alpha' + \int_{\partial\Sigma} \beta'. \]
	\end{lem}

	Not every element in $\cala^{1,1}(\Sigma,*D)$ is integrable, but by Lemma \ref{Regularized Integral}, we can define the \emph{regularized integral} on $\cala^{1,1}(\Sigma,*D)$ by
	\[\Xint-_{\Sigma}\omega := \int_{\Sigma}\alpha + \int_{\partial \Sigma} \beta,\]
	with the commutative diagram 
	\[\begin{tikzcd}
		\cala^{1,1}(\Sigma,\log D) \arrow[rr, hook] \arrow[ddr,"\int_{\Sigma}"swap]& & \cala^{1,1}(\Sigma,*D) \arrow[ldd, "\Xint-_{\Sigma}"] \\
		& & \\
		 & \bbc & 
	\end{tikzcd}.\]

	\begin{prop} \label{Properties of regularized integral}
		We list some important properties of regularized integral here, which will be used in following sections.
		\begin{itemize}
			\item[$1.$] Let $\Sigma$ be a compact Riemann surface without boundary, and $D \subset \Sigma$ be a finite subset. Let $\xi \in \cala^{1,0}(\Sigma,*D)$ and $\eta \in \cala^{0,1}(\Sigma,*D)$. Then
		\[\Xint-_{\Sigma} \bar{\partial} \xi = -2\pi i \op{Res}_{\Sigma}(\xi), \quad \Xint-_{\Sigma} \partial \eta = 0.\]
		Here
		\[\op{Res}_{\Sigma}(\xi) = \sum_{p \in D}\op{Res}_{p}(\xi), \quad \text{with} \quad  \op{Res}_{p}(\xi) = \frac{1}{2\pi i} \lim_{\varepsilon \to 0}\int_{|z|=\varepsilon}\xi,\] 
		and $z$ is a local holomorphic coordinate around $p$ with $z(p)=0$.
		
		\item[$2.$] If $f$ is smooth around $0 \in \bbc$, and $n \in \bbn_{+}$, then
		\[\op{Res}_{z=0}\kuohao{\frac{f(z)}{z^{n}}dz} = \frac{1}{(n-1)!} \partial_{z}^{n-1}f(0).\]
		If $f$ is a smooth function, with a holomorphic pole on $0$, and $l \in \bbn_{+}$, then
		\[\op{Res}_{z=0}\kuohao{\bar{z}^{l}f(z)dz} = 0.\]
				
		\item[$3.$] The regularized integral can be generalized to integrals on the configuration spaces $\op{Conf_n}(\Sigma) = \Sigma^n - \Delta_n$, where $\Delta_n$ is the big diagonal, and one has
		\[\Xint-_{\Sigma} \colon \cala^{n,n}(\Sigma^n, *\Delta_n) \to \cala^{n-1,n-1}(\Sigma^{n-1},*\Delta_{n-1}).\] 
		\item[$4.$] The iterated integral
		\[\Xint-_{\Sigma_{\sigma(1)}} \cdots \Xint-_{\Sigma_{\sigma(n)}} \colon \cala^{n,n}(\Sigma^n,*\Delta_n) \to \bbc, \quad \sigma \in S_n\]
		does not depend on the choice of $\sigma$. Here $S_n$ is the permutation group.
		\end{itemize}
	\end{prop}
	
	\paragraph{Modularity of regularized integrals} Let $\Phi(z;\tau)$ be a function on $\bbc^{n} \times \bbh$, where $\bbh = \{\tau \in \bbc \colon \op{Im}\tau > 0\}$ is the upper half plane. The $\op{SL}(2,\bbz)$ action on $\bbc^n \times \bbh$ is given by
	\[\gamma(z;\tau) = (\gamma z;\gamma\tau) := \kuohao{\frac{z}{c\tau + d};\frac{a\tau+b}{c\tau+d}},\]
	where $\gamma = \begin{pmatrix}
		a & b\\
		c & d
	\end{pmatrix} \in  \op{SL}(2,\bbz).$ 
	\begin{itemize}
		\item $\Phi$ is \emph{elliptic} if for any $\lambda \in \Lambda_{\tau} = \{m + n\tau \mid m,n \in \bbz\}$,
	\[\Phi(z+\lambda;\tau) = \Phi(z;\tau).\] 
	\item $\Phi$ is \emph{modular with weight $k \in \bbz$} if 
	\[\Phi(\gamma z;\gamma \tau) = (c\tau+d)^{k}\Phi(z;\gamma), \quad \forall \gamma \in \op{SL}(2,\bbz).\]
	\end{itemize}
	
	If $\Phi$ is an elliptic function, then for each $\tau \in \bbh$, $\Phi(z;\tau)$ is well-defined on $E_{\tau}^{n}$, so we can consider the function 
	\[f(\tau) = \Xint-_{E_{\tau}^{n}} \Phi(z;\tau) \op{vol},\]
	where 
	\[\op{vol} = \prod_{j=1}^{n} \frac{i}{2 \op{Im}\tau}dz_j \wedge d\bar{z}_j\]
	is the volume form on $E_{\tau}^{n}$ with
	\[\int_{E_{\tau}} \frac{i}{2\op{Im}\tau}dz \wedge d\bar{z} = 1.\]
	\begin{thm}[= Theorem 3.4 in \cite{li2021regularized}]\label{thm about modularity}
	Let $\Phi(z;\tau)$ be a meromorphic elliptic function on $\bbc^n \times \bbh$ which is holomorphic away from diagonals and modular of weight $k$. Then
	\[\Xint-_{E_{\tau}^{n}} \Phi \op{vol} \in \bbc[E_4,E_6][\widehat{E}_2]\] 
	is modular of weight $k$ and almost holomorphic on $\bbh$, where $E_{k},k \geq 2$ are Eisenstein series, and $\widehat{E}_2 = E_2 - \frac{3}{\pi \op{Im}\tau}$.
	\end{thm}
	
	\paragraph{Regularized integrals on elliptic functions} Let $\wp$ be the Weierstrass's elliptic function, where
	\[\wp(z;\tau) = \frac{1}{z^2} + \sum_{\lambda \in \Lambda_{\tau}\setminus\{0\}} \left(\frac{1}{(z+\lambda)^2} - \frac{1}{\lambda^2}\right).\]
	Let 
	\[\zeta(z;\tau) = \frac{1}{z} + \sum_{\lambda \in \Lambda_{\tau}\setminus\{0\}} \left(\frac{1}{z+\lambda} - \frac{1}{\lambda} + \frac{z}{\lambda^2}\right),\]
	then $\wp$ and $\zeta$ are modular with weight $2$ and $1$ respectively, with $\wp = -\zeta'$. Since $\wp$ is elliptic, 
	\[\zeta(z+1;\tau) - \zeta(z;\tau) = \omega_1, \quad \zeta(z+\tau;\tau) - \zeta(z;\tau) = \omega_2\]
	are constants. Let $P = \{\mu_1 + \mu_2\tau \mid -\frac{1}{2} \leq \mu_1,\mu_2 \leq \frac{1}{2}\}$ be a fundamental domain \cite{martin1966complex}. Use the relation
	\[\frac{1}{2\pi i}\int_{\partial P}\zeta dz = 1,\]
	one has the equation
	\[\omega_1 \tau - \omega_2 = 2\pi i.\]	 
	Moreover,
	\[\zeta(z+1;\tau) - \zeta(z;\tau) = \frac{\pi^2}{3}E_{2}(\tau).\]
	Now consider the following functions
	\begin{align*}
		Z(z) &:= \zeta(z) - \frac{\pi^2}{3}E_2(\tau)z,\\
		\widehat{Z}(z,\bar{z}) &:= \zeta(z) - \frac{\pi^2}{3}E_2(\tau)z + \A, \quad \A := - \frac{\pi}{\op{Im}\tau}(\bar{z} - z)\\
		\widehat{P}(z) &:= - \partial_z \widehat{Z} = \wp(z) + \frac{\pi^2}{3}E_2(\tau) + \Y, \quad \Y := -\frac{\pi}{\op{Im}\tau},
	\end{align*}
	then $\widehat{Z}$ and $\widehat{P}$ are elliptic and modular with weight $1$ and $2$ respectively. 
	
	\begin{prop}\label{Laurent' series 1}
		The Laurent's series of $\xi$ and $\wp$ at $0$ are	
	\begin{align*}
		\zeta(z) &=  \frac{1}{z} - \sum_{k=3}^{\infty} G_{k+1}z^{k} , \\
		 \wp(z) &= \frac{1}{z^2} + \sum_{k=2}^{\infty}(k+1)G_{k+2}z^{k},
	\end{align*}
	where 
	\[G_k = 2 E_k\cdot\sum_{n=1}^{\infty}\frac{1}{n^k} = \sum\limits_{\lambda \in \Lambda\setminus\{0\}} \frac{1}{\lambda^{k}}.\] 
	\end{prop}
	\begin{proof}
		Direct computation.
	\end{proof}
	
	\begin{lem}[=Lemma 2.6 in \cite{li2023regularized}]\label{Residual formulas for Z-hat}
	 	Let $\psi$ be a meromorphic function on $E_\tau$, and $k \geq 0$ be an integer. One has
	 	\[\Xint-_{E_\tau} \psi \widehat{Z}^{k} \op{vol} = \op{Res}\kuohao{\psi \frac{\widehat{Z}^{k+1}}{k+1}dz}.\]
	\end{lem}
	
	\begin{eg} \label{Value of P}
		Applying Lemma \ref{Residual formulas for Z-hat} and Proposition \ref{Laurent' series 1}, one has
		\begin{align*}
		\Xint-_{E_\tau}\widehat{P}\op{vol} &= \op{Res}_{z=0}(\widehat{P}\widehat{Z}dz) = 0,\\
		\Xint-_{E_\tau}\widehat{P}^2\op{vol} &= \op{Res}_{z=0}(\widehat{P}^2\widehat{Z}dz) =\pi^4\frac{-\widehat{E}^2_2+E_4}{9},\\
		\Xint-_{E_\tau}\widehat{P}^3\op{vol} &= \op{Res}_{z=0}(\widehat{P}^3\widehat{Z}dz) = \pi^6\frac{-10\widehat{E}_2^3+6\widehat{E}_2E_4+4E_6}{5\cdot 27}.
	\end{align*}
	\end{eg}

	\paragraph{Holomorphic anomaly equation} Let $\mathcal{F}_n$ be the set of functions on $\bbc^n \times \bbh$ that are smooth everywhere, except possibly with holomorphic poles along
	\[D_{n} = \{(z_1,\cdots,z_n;\tau) \mid \exists 1 \leq i < j \leq n,\  s.t.\  z_i-z_j \in \Lambda_{\tau}\}.\]
	Let $\mathcal{J}_n = \mathcal{F}_n \cap \bbc(\wp_{ij},\wp_{ij}',E_4,E_6)$, and $\widehat{\mathcal{J}}_{n} = \mathcal{J}_n[\widehat{Z}_{ij},\widehat{E}_2]$, where $f_{ij} = f(z_i-z_j)$ for $f $ taking $ \wp,\wp',\widehat{Z},\A$. 
	\begin{thm}[Theorem 1.2 in \cite{li2023regularized}]\label{Si Li's main theorem}
		Let $\Psi \in \widehat{J}_n, n\geq 2$ be an almost-elliptic function, then one has
		\[\partial_{\Y} \Xint-_{E_{\tau}^n} \Psi = \Xint-_{E_{\tau}^n} \partial_{\Y}\Psi - \sum_{a<b} \Xint-_{E_{\tau}^{n-1}}\op{Res}_{z_a=z_b}\kuohao{(z_a-z_b)\Psi dz_a}.\]
		Here in computing $\partial_{\Y} \Psi$, $\Psi$ is regarded as an element in $\mathcal{J}_n[Z_{ij},E_{2}] \otimes \bbc[\A_{ij},\Y]$ with the convention $\partial_{\Y}\A_{ij}=0$.
	\end{thm}	
	
	\begin{rmk}
		In fact, for $a \neq b$, one has the equality
		\[\op{Res}_{z_a=z_b}\kuohao{(z_a-z_b)\Psi dz_a} = \op{Res}_{z_b=z_a}\kuohao{(z_b-z_a)\Psi dz_b}.\]
	\end{rmk}

%% file: Regularized_Integral_on_Graphs.tex
\section{Regularized Integral on Riemann Graph}

	\subsection{Regularized Integral on Feymann Graphs}
	
	\begin{defn} \label{Def of decorated graph}
		A decorated graph $(\Gamma,n)$ is the following set of data:
		\begin{itemize}
			\item A set $V(\Gamma)$ of vertices.
			\item A set $E(\Gamma)$ of directed \emph{edges} and a set $L(\Gamma)$ of \emph{loops}.
			 \item The assignment of head and tail to each edge, i.e. two maps $h,t \colon E(\Gamma) \to V(\Gamma)$ with $h(e) \neq t(e)$, $\forall e \in E$.
			 \item The assignment of vertex to each loop, i.e. a map $\mu \colon L(\Gamma) \to V(\Gamma)$.
			 \item The assignment of decorations, i.e. a map $n \colon E(\Gamma) \cup L(\Gamma) \to \bbn$.
		\end{itemize}
	\end{defn}
	We will only consider \emph{finite} decorated graphs, i.e. with $V(\Gamma)$, $E(\Gamma)$ and $L(\Gamma)$ finite. We will always denote $m$ as the number of vertices in $\Gamma$. 
	\begin{eg}
		Here is a decorated graph with two vertices $v,w$, two edges with decoration $0$ and one loop with decoration $2$. 
		\[
			\begin{tikzpicture}
				\node[point] (a) [label = below:{$v$}]{};
				\node[point] (b) at (4,0) [label = below:{$w$}] {};
				\path[-latex] (a) edge[bend left] node[above] {0} (b);
				\path[-latex] (a) edge[bend right] node[below] {0} (b);
				\path[-latex] (a) edge[reflexive = 45, in = 225, out=135] node[left] {2} (a);
			\end{tikzpicture}
		\]

	\end{eg}
		
	\begin{defn}
		Let $(\Gamma,n)$ be a decorated graph and let $E_{\tau}$ be an elliptic curve. Let 
		\[\Phi_{(\Gamma,n)} = \prod_{e \in E(\Gamma)}\pds{z_{h(e)}}{n_e}\widehat{P}(z_{h(e)}-z_{t(e)}),\]
		and let 
		\[\widetilde{W}(\Gamma,n) := \Xint-_{E_{\tau}^{m}} \Phi_{(\Gamma,n)} \op{vol},\]
		where $m$ is the number of vertices in $\Gamma$.
	\end{defn}
	\begin{rmk}
		The information of loops are not included here. We will treat the contribution of loops in Lemma/Definition \ref{Contribution of loops}.
	\end{rmk}
	\begin{defn}
		Let $(\Gamma,n)$ be a decorated graph. For $A \subset E \cup L$, the \emph{weight of $A$} is defined as 
		\[\omega(A) := \sum_{e \in A} (n_e+2).\]
	\end{defn}
	\begin{rmk}
		The function $\Phi_{(\Gamma,n)}$ is holomorphic away from diagonals and modular of weight $\omega(E)$, and $\widetilde{W}(\Gamma,n) \in \bbc[E_4,E_6][\widehat{E}_2]$ is modular of weight $\omega(E)$ and almost holomorphic on the upper half plane $\bbh$.
	\end{rmk} 
	
	If $e$ is a loop attached on $w$ with decoration $k$, then it can be considered as a limit process shown as follows,
	\[
		\begin{tikzpicture}[modal]
			\node[point] (w) [label=below:{$w$}] {};
			\node[point] (v) at (-1,0) [label=below:{$v$}] {};
			\path[-latex] (v) edge[reflexive = 5, in = 45, out=135] node[above] {$k$} (w);
			\path[->] (v) edge[dashed] (w);
		\end{tikzpicture}
	\]
	which leads us to consider the local information around each point $z_v \in E_{\tau}$.
	
	\begin{lem-defn} \label{Contribution of loops}
		The regularized integral on loops with decoration $k$ is defined as follows:
		\[W_k = \lim_{\varepsilon \to 0} \frac{\Xint-_{E_{\tau} \times B_{\varepsilon}(z_v)}\pds{z_v}{k}\widehat{P}(z_v-z_w)\op{vol}}{\Xint-_{E_{\tau} \times B_{\varepsilon}(v_1)} \op{vol}}.\]
		Moreover, the limit exists and one has
		\[W_k = \begin{cases}
		\frac{\pi^2}{3}\widehat{E}_2, \quad &k=0\\
		(k+1)!G_{k+2}, \quad & k>0.
	\end{cases}\]
	\end{lem-defn}
	\begin{proof}
		Notice that
		\[\pds{z_v}{k}\widehat{P}(z_v-z_w)\vol{w} = \partial_{w} \kuohao{\pds{z_v}{k}\widehat{Z}(z_v-z_w)\frac{d\bar{z}_w }{2i\op{Im}\tau}}.\]
		By definition of regularized integral, one has
		\[\Xint-_{B_{\varepsilon}(z_v)}\pds{z_v}{k}\widehat{P}(z_v-z_w)\vol{w} = \int_{\partial B_{\varepsilon}(z_v)}\pds{z_v}{k}\widehat{Z}(z_v-z_w)\frac{d\bar{z}_w }{2i\op{Im}\tau}.\]
		Let $z_w = z_v - \varepsilon e^{i\theta}$, then
		\[\int_{\partial B_{\varepsilon}(z_v)}\pds{z_v}{k}\widehat{Z}(z_v-z_w) d\bar{z}_w = \int_{0}^{2\pi} \partial^{k}\widehat{Z}(\varepsilon e^{i\theta}) \cdot (i\varepsilon e^{-i\theta})d\theta.\]
		The value of above integration depends on the Fourier coefficient of $e^{i\theta}$ about the function $\partial^{k}\widehat{Z}$, which is only associated to the coefficient of $z$ in this case. Therefore,  
		\[\int_{0}^{2\pi} \partial^{k}\widehat{Z}(\varepsilon e^{i\theta}) \cdot (i\varepsilon e^{-i\theta})d\theta = \begin{cases}
		-2\pi i \varepsilon^2 \frac{\pi^2}{3}\widehat{E}_2, \quad &k=0\\
		-2\pi i \varepsilon^2 (k+1)!G_{k+2}, \quad & k>0.
		\end{cases}\]
		Then it's easy to obtain
		\[W_k = \begin{cases}
			\frac{\pi^2}{3}\widehat{E}_2, \quad &k=0\\
			(k+1)!G_{k+2}, \quad & k>0,
		\end{cases}\]
		which completes the proof.
	\end{proof}

	\begin{eg} \label{The value of Wk's}
		For all odd integers, we have $W_k = 0$. And 
		$W_0 = \frac{\pi^2}{3}\widehat{E}_2$, 
		\[W_2 = 3! \cdot \frac{\pi^4}{90} \cdot 2E_4 = \frac{2\pi^4}{15}E_4, \quad W_4 = 5! \cdot \frac{\pi^6}{945} \cdot 2E_6 = \frac{16\pi^6}{63}E_6.\]
	\end{eg}
	
	\begin{cor} \label{loop formula of propagator}
		Using Proposition \ref{Laurent' series 1}, the Laurent's series of $\widehat{P}$ at $0$ is 
		\[\widehat{P}(z) = \frac{1}{z^2} + \sum_{k=0}^{\infty} \frac{W_k}{k!}z^k.\]
		Similarly, the  Laurent's series of $\partial^{n}\widehat{P}$ at $0$ is 
		\[\partial^{n}\widehat{P}(z) = (-1)^n \frac{(n+1)!}{z^{n+2}} + \sum_{k=0}^{\infty} \frac{W_{n+k}}{k!}z^k.\]
 	\end{cor}
 	
 	\begin{defn} \label{Definition of R.I}
 		For a decorated graph $(\Gamma,n)$, let 
 		\[\Psi_{(\Gamma,n)} := \Phi_{(\Gamma,n)} \prod_{e \in L} W_{n_e}.\]
 		\emph{The regularized integral} of $(\Gamma,n)$ on $E_{\tau}$ is defined by
 		\[W(\Gamma,n) : = \widetilde{W}(\Gamma,n) \prod_{e \in L}W_{n_e},\]
 		or similarly
 		\[W(\Gamma,n) := \Xint-_{E_{\tau}^m} \Psi_{(\Gamma,n)} \op{vol}.\]
 	\end{defn}
 	
 	\begin{eg}
		If the decorated graph $(\Gamma,n)$ is given as follows,
		\[
			\begin{tikzpicture}
				\node[point] (a) [label = below:{$v$}]{};
				\node[point] (b) at (4,0) [label = below:{$w$}] {};
				\path[-latex] (a) edge[bend left] node[above] {0} (b);
				\path[-latex] (a) edge[bend right] node[below] {0} (b);
				\path[-latex] (a) edge[reflexive = 45, in = 225, out=135] node[left] {0} (a);
			\end{tikzpicture}
		\]
		then $\Phi_{(\Gamma,n)} = \widehat{P}^2(z_v-z_w)$. By Example \ref{Value of P}, 
		\[\widetilde{W}(\Gamma,n) = \pi^4\frac{-\widehat{E}^2_2+E_4}{9} \quad W(\Gamma,n) = \pi^6 \frac{-\widehat{E}_2^3+E_4\widehat{E}_2}{3^3}.\]
	\end{eg}
	
	\begin{eg}
		Let $(\Gamma,n)$ be a decorated graph. If there exists a vertex $v \in V(\Gamma)$ such that there is only one edge connects $v$, say $e$, then
		\[\Phi_{(\Gamma,n)} = \pds{z_{h(e)}}{n_e} \widehat{P}(z_{h(e)}-z_{t(e)}) \cdot \Theta,\]
		where $\Theta$ does not have the variable $z_v$. Since $\partial^{n_e}_{z}\widehat{P} = \partial^{n_e+1}_{z}\widehat{Z}$, by the first part of Proposition \ref{Properties of regularized integral}, one has
		\[\Xint-_{E_{\tau}} \Phi_{(\Gamma,n)} \vol{v} = 0.\]
		So $W(\Gamma,n) = 0$.
	\end{eg}
	
	By Theorem \ref{Si Li's main theorem}, the anti-holomorphic dependence of $\widetilde{W}(\Gamma,n)$ is given by
	\[\partial_{\Y} \widetilde{W}_{(\Gamma,n)} = \Xint-_{E_{\tau}^m} \partial_{\Y}\Phi_{(\Gamma,n)} \op{vol} - \sum_{a < b} \op{Res}_{z_{v_a} = z_{v_b}}\kuohao{(z_{v_a}-z_{v_b})\Phi_{(\Gamma,n)} dz_{v_a}},\]
	which leads us to consider the geometric meanings of 
	\[\partial_{\Y}\Phi_{(\Gamma,n)} \quad \text{and} \quad \op{Res}_{z_v=z_w}\kuohao{(z_v-z_w)\Phi_{(\Gamma,n)}dz_{v}}.\]
	
	\begin{lem}\label{First part of differential}
		Let $(\Gamma,n)$ be a decorated graph. Then 
		\begin{equation} 
			\partial_{\Y}\Phi_{(\Gamma,n)} = \sum_{e \in E(\Gamma), \atop n_e = 0} \Phi_{(\Gamma \setminus e,n)},
		\end{equation}
		where $(\Gamma \setminus e,n)$ is the graph by deleting the edge $e$ in $\Gamma$, and the decoration on $\Gamma\setminus e$ is the restriction of the decoration on $\Gamma$.
	\end{lem}
	\begin{proof}
		Since $\partial_{\Y} \partial^{n}\widehat{P}(z) = 0$ for $n > 0$ and $\partial_{\Y} \widehat{P}(z) = 1$. By Leibniz's rule,
		\[\partial_{\Y}\Phi_{(\Gamma,n)} = \sum_{e \in E(\Gamma), \atop n_e = 0} \Phi_{(\Gamma \setminus e,n)},\]
		which completes the proof.
	\end{proof}
	
	Now let's focus on the residual formula
	\begin{equation} 
		\op{Res}_{z_v = z_w} \kuohao{(z_v-z_w)\prod_{e \in E(\Gamma)} \pds{z_{h(e)}}{n_e}\widehat{P}(z_{h(e)} - z_{t(z)})dz_v}.
	\end{equation}
	Clearly, it only depends on the terms with variable $z_v$. Geometrically, these terms are propagators put on edges who connect the vertex $v$.
	\begin{figure}[htb]
		\[
		\begin{tikzpicture}[modal]
			\node[point] (w) [label=below:{$w$}] {};
			\node[point] (v) at (4,0) [label=below:{$v$}] {};
			\path[-latex] (v) edge[bend right] node[above] {$n_1$} node[below]{$\vdots$}(w);
			\path[-latex] (v) edge[bend left] node[below] {$n_k$} (w);
			\path[->] (v) edge[] node[above] {$m_1$} node[below]{$\ \vdots$} (6,1);
			\path[->] (v) edge[] node[below] {$m_r$} (6,-1);
		\end{tikzpicture}
		\]
		\caption{The edges connecting $v$ in $\Gamma$}
		\label{figure 1}
	\end{figure}
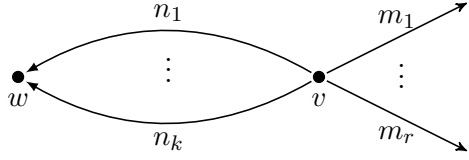
	
	Without loss of generality, suppose there are $k$ edges connecting $v$ and $w$, with decorations $n_1,\cdots,n_k$. And the other $r$ edges connecting $v$ are $e_1,\cdots,e_r$, with decorations $m_1,\cdots,m_r$. Moreover, suppose $v$ is the head of those edges, and the graph is shown in Figure \ref{figure 1}. 
	
	\begin{lem} \label{Summation of residual graph terms}
		The residual formula 
		\begin{equation} 
			\op{Res}_{z_v=z_w} \kuohao{ (z_v-z_w) \cdot \prod_{i=1}^{k} \pds{z_v}{n_i}\widehat{P}(z_v-z_w) \cdot \prod_{i=1}^{r} \pds{z_{v}}{m_i}\widehat{P}(z_{v}-z_{t(e_i)})dz_v}
		\end{equation}
		is the summation of terms like
		\begin{equation}\label{Residual graph}
			W_{j_{s+1}+n_{s+1}} \cdots W_{j_k+n_k}\prod_{i=1}^{r} \pds{z}{m_i+u_i}\widehat{P}(z_w -z_{t(e_i)}),
		\end{equation}
		with coefficients 
		\[(-1)^{n_1+\cdots+n_s}\frac{(n_1+1)!\cdots(n_s+1)!}{j_{s+1}!\cdots j_{k}!u_1!\cdots u_r!},\]
		where $n_1+\cdots+n_s + 2s -2 = j_{s+1} + \cdots + j_k + u_1 + \cdots + u_r.$
	\end{lem}
	
	\begin{rmk}
		Formula \eqref{Residual graph} can be given by the decorated graph in Figure \ref{figure 2} , which is a graph by collapsing the first $s$ edges between $v$ and $w$ in Figure \ref{figure 1}. 
	Roughly speaking, the weights of those $s$ edges are assigned to other edges in Figure \ref{figure 1}.
	\end{rmk}
	
	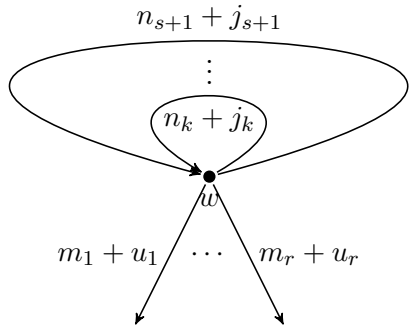
\begin{figure}[htb]
		\[\begin{tikzpicture}[modal]
			\node[point] (v) [label=below:{$w$}] {};
			\path[->] (v) edge[reflexive=50,in=150,out=30] node[below]{$n_k+j_k$} (v);
			\path[->] (v) edge[reflexive=150,in=165,out=15] node[above]{$n_{s+1}+j_{s+1}$} (v);
			\path[->] (v) edge[] node[left]{$m_1+u_1$} node[right]{$\ \cdots$}(-1,-2);
			\path[->] (v) edge[] node[right]{$m_{r}+u_{r}$} (1,-2);
			\node (dot) at (0,1.5) {$\vdots$};
		\end{tikzpicture}\]
		\caption{Quotient graph of Figure \ref{figure 1}}
		\label{figure 2}
	\end{figure}

	\begin{proof}[Proof of Lemma \ref{Summation of residual graph terms}]
		Let $z = z_v - z_w$, by Corollary \ref{loop formula of propagator}, one has
		\begin{align*}
			&\op{Res}_{z_v=z_w} \kuohao{ (z_v-z_w) \cdot \prod_{i=1}^{k} \pds{z_v}{n_i}\widehat{P}(z_v-z_w) \cdot \prod_{i=1}^{r} \pds{z_{v}}{m_i}\widehat{P}(z_{v}-z_{t(e_i)})dz_v}\\
			 = & \op{Res}_{z=0}\kuohao{ z \cdot \prod_{i=1}^{k}\kuohao{(-1)^{n_i}\frac{(n_i+1)!}{z^{n_i+2}} + \sum_{j=0}^{\infty}\frac{W_{n_i+j}}{j!}z^j} \cdot \prod_{i=1}^{r} \pds{z}{m_i}\widehat{P}(z + z_w -z_{t(e_i)})dz}.
		\end{align*}
		Since
		\begin{equation} \label{Other edges connecting v}
			\prod_{i=1}^{r} \pds{z}{m_i}\widehat{P}(z + z_w -z_{t(e_i)}).
		\end{equation}
		is holomorphic at $z=0$. By Proposition \ref{Properties of regularized integral} , the residual formula can be computed by
		\[D_{v,w} \kuohao{\prod_{i=1}^{r} \pds{z}{m_i}\widehat{P}(z + z_w -z_{t(e_i)})} \bigg|_{z=0},\] 
		where $D_{v,w}$ is a differential operator determined by the singular part of 
		\begin{equation}\label{Residual formula w.r.t v and w}
			z \cdot \prod_{i=1}^{k}\kuohao{(-1)^{n_i}\frac{(n_i+1)!}{z^{n_i+2}} + \sum_{j=0}^{\infty}\frac{W_{n_i+j}}{j!}z^j}.
		\end{equation}
		If we expand the brackets in \eqref{Residual formula w.r.t v and w}, one of the terms has the following formula
	
		\begin{equation}\label{One of the term}
			(-1)^{n_1+\cdots+n_s}\frac{(n_1+1)!\cdots (n_s+1)!}{j_{s+1}! \cdots j_{k}!}W_{j_{s+1}+n_{s+1}} \cdots W_{j_k+n_k} z^{-(l +1)},
		\end{equation}
		where \[l = n_1+\cdots+n_s+2s-2-j_{s+1}-\cdots-j_k .\]
		If $l \geq 0$, then \eqref{One of the term} is a singular term and it can transform to the differential operator
		\[(-1)^{n_1+\cdots+n_s}\frac{(n_1+1)!\cdots (n_s+1)!}{j_{s+1}! \cdots j_{k}! \cdot l!}W_{j_{s+1}+n_{s+1}} \cdots W_{j_k+n_k}\partial^{l}.\]
		By Leibniz's rule,
		\[\partial^{l}\prod_{i=1}^{r} \pds{z}{m_i}\widehat{P}(z + z_w -z_{t(e_i)}) = \sum_{u_1,\cdots,u_r \in \bbn \atop u_1+\cdot+u_r=l} \frac{l!}{u_1!\cdots u_r!}\prod_{i=1}^{r} \pds{z}{m_i+u_i}\widehat{P}(z + z_w -z_{t(e_i)}).\]
		Therefore, the residual formula is the summation of terms like \eqref{Residual graph}, with coefficients 
		\[(-1)^{n_1+\cdots+n_s}\frac{(n_1+1)!\cdots(n_s+1)!}{j_{s+1}!\cdots j_{k}!u_1!\cdots u_r!},\]
		which completes the proof.
	\end{proof}

	To write the residual formula more precisely, we will introduce some necessary notations and definitions in next section.

	\subsection{$k$-shifted Residual Graph}			
	\begin{defn}
		Let $A$ be a set, and $m$ be an integer. A function 
		$u \colon A \to \bbn$
		is an assignment of $m$ to $A$ if 
		\[\sum_{a \in A} u(a) = m.\]
		And the set of assignments of $m$ to $A$ is denoted by $P(A, m)$.
	\end{defn}
		
	\begin{defn}
		Let $A \subset E(\Gamma),E \subset E(\Gamma) \cup L(\Gamma)$ be two non-empty edge sets of a decorated graph $(\Gamma,n)$. For $u \colon E \to \bbn$, $(\Gamma, n+u)$ is a decorated graph with a new decoration 
		\[(n+u)(e) = \begin{cases}
			n(e) \quad &e \notin E,\\
			n(e) + u(e) \quad & e \in E.
		\end{cases}\]
		And $(\Gamma/A, n)$ is a quotient graph given by $\Gamma$ collapsing $A$. The decoration on $\Gamma/A$ is a restriction map
		\[n \colon (E(\Gamma) \cup L(\Gamma)) \setminus A \to \bbn.\]
	\end{defn}
	
	\begin{defn}
		Let $v, w \in V(\Gamma)$ be two vertices in a decorated graph $(\Gamma,n)$. Let
		\[E^{+}_{v} = \{e \in E(\Gamma) \mid h(e) = v\}\]
		and 
		\[E^{-}_{v} = \{e \in E(\Gamma) \mid t(e) = v\},\]
		be sets of edges whose head (tail) is $v$. 
		
		Let $E_v = E_v^+ \cup E_v^-$ be the set of edges attaching $v$ and let $E_{v, w} = E_v \cap E_w$ be the set of edges connecting $v$ and $w$.	
	\end{defn}
	
	\begin{defn}
		Let $v, w$ be two vertices in a decorated graph $(\Gamma, n)$ with $E_{v, w} \neq \emptyset$ and let $l$ be an integer. For a nonempty set $A \subset E_{v, w}$ and an assignment function $u \in P(E_{v} \setminus A, l)$, the \emph{collapse coefficient of $\Gamma$ with respect to $A$ and $u$} is defined by
		\[C_{A,u} = (-1)^{c_{A,u}} \cdot \frac{\prod_{e\in A}(n_e+1)!}{\prod_{e \in E_v \setminus A}u_e!},\]
		where
		\[c_{A,u} = \omega(E_v^- \cap E_{v,w})+\omega(A)+\sum_{e \in E_v^- \setminus E_{w}}u_e.\]
	\end{defn}
	
	\begin{defn}
		The set of decorated graphs is denoted by $G$, and let $V(G)$ be the $\bbc$-linear space spanned by $G$.
	\end{defn}
	
	\begin{defn} \label{Definition of k-shifted residual graph}
		Let $(\Gamma,n)$ be a decorated graph, and let $v, w$ be vertices in $\Gamma$ such that $E_{v, w} \neq \emptyset$. For a nonempty set $A \subset E_{v, w}$ and a integer $k$, the \emph{$k$-shifted residual graph with respect to $A$ is defined by}
		\[\op{Res}_{w}^{(v)}[k](\Gamma,n;A) := \sum_{u \in P(E_v\setminus A, \omega(A)-1-k)} C_{A,u}\cdot(\Gamma/A,n+u).\]
		The \emph{$k$-shifted residual graph of $(\Gamma,n)$ with $v$ collapsing to $w$} is defined by
		\[\op{Res}_{w}^{(v)}[k](\Gamma,n) := \sum_{A \subset E_v \atop A \neq \emptyset} \op{Res}_{w}^{(v)}[k](\Gamma,n;A).\] 
	\end{defn}
	\begin{rmk}
		If $E_{v, w} = \emptyset$, the $k$-shifted residual graph of $(\Gamma,n)$ with $v$ collapsing to $w$ is $0 \in V(G)$.
	\end{rmk}
	
	\begin{prop}
		let $(\Gamma,n)$ be a decorated graph and let $k$ be an integer, then
		\[\Phi_{\op{Res}_w^{(v)}[k](\Gamma,n)} = (-1)^{k+1}\Phi_{\op{Res}_{v}^{(w)}[k](\Gamma,n)}.\]
	\end{prop}
	\begin{proof}
		It's straightforward to check.
	\end{proof}
	
	\begin{defn}
		Let $\gamma = \sum c_i(\Gamma_i,n_i)$ be an element in $V(G)$, then we will use $\Phi_{\gamma}$ to define the linear combination of the functions $\Phi_{(\Gamma_i,n_i)}$:
		\[\Phi_{\gamma} := \sum c_i\Phi_{(\Gamma_i,n_i)}.\]
	\end{defn}
	
	\begin{lem} \label{Second part of differential}
		Let $(\Gamma,n)$ be a decorated graph, and let $v, w$ be two vertices in $\Gamma$ with $E_{v, w} \neq \emptyset$. Then 
		\[\op{Res}_{z_v =z_w}\kuohao{(z_v-z_w)\Psi_{(\Gamma,n)}dz_v} = \Psi_{\op{Res}_{w}^{(v)}[1](\Gamma,n)}.\]
	\end{lem}
	\begin{proof}
		This follows from Lemma \ref{Summation of residual graph terms}. 
	\end{proof}

	\begin{rmk}\label{Symmetry}
		It is easy to check that $\Phi_{\op{Res}_w^{(v)}[k](\Gamma,n)} = (-1)^{k+1}\Phi_{\op{Res}_{v}^{(w)}[k](\Gamma,n)}$.
	\end{rmk}

	\subsection{The Holomorphic Anomaly Equation on Graphs}
	\begin{lem-defn}
		Suppose $(\Gamma,n)$ is a decorated graph, let
		\[\delta(\Gamma,n) = \sum_{e \in E \cup L \atop n_e=0}(\Gamma \setminus e,n) - \frac{1}{2}\sum_{v,w \in V \atop v \neq w} \op{Res}_{w}^{(v)}[1](\Gamma,n).\]
		
		Then $(V(G),\cup,\delta)$ is a differential ring, where $\cup$ is the product on $V(G)$ induced by the union of graphs
		$\cup \colon G \times G \to G.$
	\end{lem-defn}
	\begin{proof}
		Obvious.
	\end{proof}
	
	\begin{thm}\label{Main Theorem}
		The following diagram commutes.
		\[\begin{tikzcd}
		V(G) \arrow[r,"W"] \arrow[d,"\delta"] &  \bbc[E_4,E_6] \arrow[d,"\partial_{\Y}"] [\widehat{E}_2]\\
		V(G) \arrow[r,"W"] & \bbc[E_4,E_6][\widehat{E}_2]
	\end{tikzcd}\]
	\end{thm}	
	
	\begin{proof}
		For a decorated graph $(\Gamma,n)$, one has
		\begin{align*}
			\partial_{\Y}W(\Gamma,n) &= \partial_{\Y} \kuohao{\widetilde{W}(\Gamma,n) \cdot \prod_{e\in L(\Gamma)}W_{n_e}}\\
			& = \prod_{e \in L}W_{n_e} \cdot \partial_{\Y}\widetilde{W}(\Gamma,n) + \partial_{\Y} \prod_{e \in L}W_{n_e} \cdot \widetilde{W}(\Gamma,n)\\
			&= \prod_{e \in L} W_{n_e} \cdot \left( \Xint-_{E_{\tau}^m} \partial_{\Y} \Phi_{(\Gamma,n)} \op{vol}  - \sum_{a < b} \op{Res}_{z_{v_a} = z_{v_b}}\kuohao{(z_{v_a}-z_{v_b})\Phi_{(\Gamma,n)} dz_{v_a}}\right) + \partial_{\Y} \prod_{e \in L}W_{n_e} \cdot \widetilde{W}(\Gamma,n),
		\end{align*}
		by Lemma \ref{First part of differential},
		\[\sum_{e \in L \cup E \atop n_e = 1} \Phi_{(\Gamma \setminus e,n)}=\prod_{e \in L} W_{n_e} \cdot \Xint-_{E_{\tau}^m} \partial_{\Y} \Phi_{(\Gamma,n)} \op{vol} + \partial_{\Y} \prod_{e \in L}W_{n_e} \cdot \widetilde{W}(\Gamma,n) ,\]
		and by Lemma \ref{Second part of differential} and Remark \ref{Symmetry}, 
		\[\frac{1}{2} \sum_{v,w \in V \atop v \neq w} \Psi_{\op{Res}_{w}^{(v)}[1](\Gamma,n)} = \prod_{e \in L} W_{n_e} \cdot \sum_{a < b} \op{Res}_{z_{v_a} = z_{v_b}}\kuohao{(z_{v_a}-z_{v_b})\Phi_{(\Gamma,n)} dz_{v_a}}.\]
		Therefore, 
		\[\partial_{\Y}W(\Gamma,n) = W(\delta(\Gamma,n)),\]
		which completes the proof.
	\end{proof}

	\begin{eg}
	Consider the decorated graph $(\Gamma,n)$ as follows
	\[
		\begin{tikzpicture}
			\node[point] (a) [label = below:{$v$}]{};
			\node[point] (b) at (4,0) [label = below:{$w$}] {};
			\path[-latex] (a) edge[bend left] node[above] {0} (b);
			\path[-latex] (a) edge[bend right] node[below] {0} (b);
		\end{tikzpicture}
	\]
	On one hand, 
	\[\partial_{\Y}W(\Gamma,n) = \frac{3}{\pi^2}\partial_{\widehat{E}_2}\kuohao{\pi^4\frac{-\widehat{E}^2_2+E_4}{9}} = -\frac{2\pi^2}{3}\widehat{E}_2.\]
	On the other hand, 
	\[ 
	\begin{tikzpicture}
		\node[] (lhs) at (-1.1,0) {$\delta (\Gamma,n)=$};
		\node[] (coe1) at (-0.1,0) {\Large{2}};
		\node[point] (a1) at (0.2,0) [label = below:{$v$}] {};
		\node[point] (b1) at (1.2,0) [label = below:{$w$}] {};
		\path[-latex] (a1) edge[] (b1);
		\node[] (decoration1) at (0.7,0.2) {\tiny{$0$}};
		\node[] (cde2) at (2.1,0) {\Large{$-\frac{1!}{0!}\binom{2}{1}$}};
		\node[point] (a2) at (3,0) [label = below:{$v$}] {};
		\path[-latex] (a2) edge[reflexive = 20, in = -45, out = 45] (a2);
		\node[] (decoration2) at (3.6,0) {\tiny{$0$}};
	\end{tikzpicture}\]
	and it is easy to compute
	\[W(\delta(\Gamma,n)) = -2W_0 = -\frac{2\pi^2}{3}\widehat{E}_2 = \partial_{\Y}W(\Gamma,n).\]
\end{eg}

\begin{eg}
		Then consider the case
		\[\delta(\begin{tikzpicture}
			\node[point] (v) {};
			\node[point] (w) at (1,0) {};
			\path[->] (v) edge[bend left] (w);
			\path[->] (v) edge[bend right] (w);
			\path[->] (v) edge[] (w);
		\end{tikzpicture})= \binom{3}{1}\begin{tikzpicture}
			\node[point] (v) {};
			\node[point] (w) at (1,0) {};
			\path[->] (v) edge[bend left] (w);
			\path[->] (v) edge[bend right] (w);
		\end{tikzpicture} -\binom{3}{1}\frac{1!}{0!0!}\begin{tikzpicture}
			\node[point] (a) {};
			\path[->] (a) edge[reflexive=25,in=35,out=145] (a);
			\path[->] (a) edge[reflexive=15,in=45,out=135] (a);
		\end{tikzpicture}-\binom{3}{2}\frac{1!1!}{2!}\begin{tikzpicture}
			\node[point] (a) {};
			\path[->] (a) edge[loop] node[above]{$2$} (a);
		\end{tikzpicture}.\]
		On other hand,
		\[\partial_{\Y}W(\begin{tikzpicture}
			\node[point] (v) {};
			\node[point] (w) at (1,0) {};
			\path[->] (v) edge[bend left] (w);
			\path[->] (v) edge[bend right] (w);
			\path[->] (v) edge[] (w);
		\end{tikzpicture}) = \partial_{\Y}\kuohao{\pi^6\frac{-10\widehat{E}_2^3+6\widehat{E}_2E_4+4E_6}{5\cdot 27}} = \pi^4\kuohao{-\frac{2}{3}\widehat{E}_2^2 + \frac{2}{15}E_4} .\]
		\[W(\delta(\begin{tikzpicture}
			\node[point] (v) {};
			\node[point] (w) at (1,0) {};
			\path[->] (v) edge[bend left] (w);
			\path[->] (v) edge[bend right] (w);
			\path[->] (v) edge[] (w);
		\end{tikzpicture})) =3W(\begin{tikzpicture}
			\node[point] (v) {};
			\node[point] (w) at (1,0) {};
			\path[->] (v) edge[bend left] (w);
			\path[->] (v) edge[bend right] (w);
		\end{tikzpicture}) -3W_0^2-\frac{3}{2}W_2 =\pi^4\frac{-\widehat{E}_2^2 + E_4}{3} -\frac{\pi^4}{3}\widehat{E}_2^2-\frac{3}{2} \cdot 3! \cdot \frac{\pi^4}{45}E_4,\]
		and it's easy to check 
		\[\partial_{\Y}W(\begin{tikzpicture}
			\node[point] (v) {};
			\node[point] (w) at (1,0) {};
			\path[->] (v) edge[bend left] (w);
			\path[->] (v) edge[bend right] (w);
			\path[->] (v) edge[] (w);
		\end{tikzpicture}) = W(\delta(\begin{tikzpicture}
			\node[point] (v) {};
			\node[point] (w) at (1,0) {};
			\path[->] (v) edge[bend left] (w);
			\path[->] (v) edge[bend right] (w);
			\path[->] (v) edge[] (w);
		\end{tikzpicture})).\]
	\end{eg}
	
	\begin{eg}[= Corollary 5.5 in \cite{li2012feynman}]
		Let $(\Gamma,n)$ be a decorated graph such that $L(\Gamma) = \emptyset$, $|E_{v,w}| \leq 1$ for any $v,w \in V(\Gamma)$, and the decoration on each edge is $0$. Then
		\[\partial_{\Y} W(\Gamma,n) = \sum_{e \in E} (W_{\Gamma \setminus e} - W_{\Gamma /e}),\]
		where $\Gamma/e$ is the graph by collapsing the edge $e$ in $\Gamma$ and $\Gamma \setminus e$ is the graph by deleting the edge $e$ in $\Gamma$.
	\end{eg}
	\begin{proof}
		It follows from Theorem \ref{Main Theorem}.
	\end{proof}

%% file: Regularized_Integral.tex
\section{Application: Algorithm to Compute Graph Integral}

We will use iterated integral to give a general algorithm to compute the graph integral. In this section, the decoration of a graph $\Gamma$ (See Definition \ref{Def of decorated graph}) is given by
\[n \colon E(\Gamma) \cup L(\Gamma) \to \bbn \cup \{-1\}.\]

We denote 
\begin{equation} \label{decoration -1}
	\pds{z_{h(e)}}{-1}\widehat{P}(z_{h(e)} - z_{t(e)}) : = -\widehat{Z}(z_{h(e)}-z_{t(e)}),
\end{equation}

\[\Phi_{(\Gamma,n)} = \prod_{e \in E(\Gamma)} \pds{z_{h(e)}}{n_e} \widehat{P}(z_{h(e)}-z_{t(e)}), \quad \Psi_{(\Gamma,n)} = \Phi_{(\Gamma,n)} \prod_{e \in L(\Gamma)}W_{n_e}.\]

\begin{prop}
	The Laurent's series of $-\widehat{Z}$ at $0$ is
	\[\partial^{-1}\widehat{P}(z) = -\frac{1}{z} - \frac{\pi}{\op{Im}\tau} \bar{z} + \sum_{k=0}^{\infty} \frac{W_{k-1}}{k!}z^k,\]
	where $W_{-1} = 0$.
\end{prop}

\begin{defn}
	Let $\widetilde{G}$ be the set of decoration graphs that admitting decoration $-1$, and let $V(\widetilde{G})$ be the $\bbc$-linear space spanned by $\widetilde{G}$. We define the \emph{ anti-holomorphic differential with respect to the vertex $v$ } as follows:
	\[\bar\delta_{v}(\Gamma,n) = \sum_{e \in E_v^+ \atop n_e=-1}(\Gamma \setminus e,n)- \sum_{e \in E_v^- \atop n_e = -1} (\Gamma \setminus e,n).\]
\end{defn}

In fact, $\bar\delta_{v}$ is the graph version of $\partial_{z_v}$.
\begin{prop}
	Let $(\Gamma,n)$ be a decorated graph, and $v \in E(\Gamma)$, then 
		\[\pd{\bar{z}_v}\Phi_{(\Gamma,n)} =  -\frac{\pi}{\op{Im}\tau} \Phi_{\bar{\delta}_v(\Gamma,n)}\]
\end{prop}
\begin{proof}
	This follows from Leibniz's rule and the fact that 
	\[\pd{\bar{z}_v} \widehat{Z}(z_v-z_w) = -\frac{\pi}{\op{Im}\tau}.\]
\end{proof}

	\begin{cor}
		Let $(\Gamma,n)$ be a decorated graph, and $v \in E(\Gamma)$, then 
		\[\Xint-_{E_{\tau}} \Phi_{\bar{\delta}_v(\Gamma,n)} \vol{v} = \op{Res}\kuohao{\Phi_{(\Gamma,n)}dz_v}.\]
	\end{cor}
	\begin{proof}
		Since
		\[\bar{\partial}_{z_v} \kuohao{\Phi_{(\Gamma,n)}dz_v} = \frac{\pi}{\op{Im}\tau} \Phi_{\bar{\delta}_v (\Gamma,n)} dz_v \wedge d\bar{z}_v,\]
		by Lemma \ref{Residual formulas for Z-hat}, we have
		\[\Xint-_{E_{\tau}}  \Phi_{\bar{\delta}_v(\Gamma,n)} \vol{v} = -\frac{1}{2\pi i} \Xint-_{E_{\tau}}\bar{\partial}_{z_v} \kuohao{\Phi_{(\Gamma,n)}dz_v} = \op{Res}\kuohao{\Phi_{(\Gamma,n)}dz_v},\]
		which completes the proof.
	\end{proof}
	
	For $v,w \in V(\Gamma)$, let 
	$\widehat{Z}_{v,w} : = \widehat{Z}(z_v-z_w),$
	then $\bar{\partial}_{z_v}$ is a differential on the polynomial ring $R$ generated by $\widehat{Z}_{v,w}$.    

	\begin{lem} \label{delta bar inverse}
		The differential map $\bar{\partial}_{z_v} \colon R \to R$ is surjective.
	\end{lem}
	\begin{proof}
		We only need to prove this Lemma for all monomials. Fix an element $w \in E_v$, any monomial has the formula
		\[\widehat{Z}^{l} \cdot f \quad\text{where\ } \widehat{Z} = \widehat{Z}_{v,w} \text{\ and \ } f=\prod_{u \in E_v \atop u \neq w} Z_{v,u}^{l_{u}}.  \]
		In fact, let's consider the element 
		\[g=\sum_{k=1}^{\infty} \frac{(-1)^{k-1}\widehat{Z}^{l+k}}{(l+1) \cdots (l+k)}\bar\partial_{z_v}^{k-1}f,\]
		then 
		\[\bar\partial_{z_v}g = \widehat{Z}^{l}f + \sum_{k=1}^{\infty} \kuohao{\frac{(-1)^{k-1}\widehat{Z}^{l+k}}{(l+1) \cdots (l+k)}\bar\partial_{z_v}^{k}f + \frac{(-1)^{k}\widehat{Z}^{l+k}}{(l+1) \cdots (l+k)}\bar\partial_{z_v}^{k}f } = \widehat{Z}^lf.\]
		Moreover, since the degree of $f$ is finite, then $g$ is a finite sum, thus $g \in R$, which completes the proof.
	\end{proof}
	\begin{cor}\label{Graph delta bar inverse Cor}
		Let $v$ be a vertex in a decorated graph $(\Gamma,n)$. There exists an element $\gamma$ in $V(\widetilde{G})$, such that 
		\[\bar\delta_{v} \gamma = (\Gamma,n).\]
	\end{cor}
	\begin{rmk}
		We will use the notation $\bar\delta_{v}^{-1}(\Gamma,n)$ to denote the element $\gamma$ in Corollary \ref{Graph delta bar inverse Cor}. Such $\gamma$ is not unique, but we only care about how to find one $\gamma$. And the method is given in Lemma \ref{delta bar inverse}.
	\end{rmk}
	
	\begin{cor} \label{Corollary in Section 4}
		Let $(\Gamma,n)$ be a decorated graph, and let $v$ be a vertex in $\Gamma$. Then
		\[\Xint-_{E_{\tau}} \Phi_{(\Gamma,n)} \vol{v} = \op{Res} \kuohao{\Phi_{\bar\delta^{-1}_{v}(\Gamma,n)}dz_v}.\]
	\end{cor}
	
	\begin{lem} \label{0-shifted residual graph lemma}
		Let $(\Gamma,n)$ be a decorated graph, then
		\[\op{Res}(\Psi_{(\Gamma,n)}dz_v) = \sum_{w \in N(v)}\Psi_{\op{Res}_{w}^{(v)}[0](\Gamma,n)},\]
		where $N(v) = \{ w \in V(\Gamma) \mid E_{v,w} \neq \emptyset \}$.
	\end{lem}
	\begin{proof}
		The proof of this Lemma is same as the proof of Lemma \ref{Second part of differential}.
	\end{proof}
	
	\begin{cor}
		Let $(\Gamma,n)$ be a decorated graph, and let $v$ be a vertex in $\Gamma$. Then
		\[\Xint-_{E_{\tau}} \Psi_{(\Gamma,n)} \vol{v} = \sum_{w \in N(v)}\Psi_{\op{Res}_w^{(v)}[0] \bar(\delta_{v}^{-1}(\Gamma,n))}.\]
	\end{cor}
	\begin{proof}
		It follows from Lemma \ref{0-shifted residual graph lemma} and Corollary \ref{Corollary in Section 4}.
	\end{proof}
	
	\begin{thm} \label{Main Thm 2}
		Let $(\Gamma,n)$ be a decorated graph, and let $v$ be a vertex in $\Gamma$. Then
		\[W(\Gamma,n) = \sum_{w \in N(v)} W(\op{Res}_{w}^{(v)}[0](\bar\delta^{-1}(\Gamma,n))).\]
	\end{thm} 
	\begin{proof}
		Obvious.
	\end{proof}
	
	We can use Theorem \ref{Main Thm 2} to compute the graph integral. Here are some examples:
	
	\begin{eg}
	Let us compute some graph integral 
		\[\begin{tikzpicture}
			\node[point] (a) at (-1.5,0) [label = below:{$v$}]{};
			\node[point] (b) at (1.5,0) [label = below:{$w$}]{}; 
			\path[-latex] (a) edge [in = 135, out=45] node [above] {0} (b);
			\path[-latex] (a) edge [in = -135, out=-45] node [below] {0} (b);
			\node[] (c) at (2,0) [label = above:{$\bar\delta^{-1}_{v}$}] {$\longrightarrow$};
		\end{tikzpicture} 
		\begin{tikzpicture}
			\node[] (c) at (-2,0) {\LARGE{$-$}};
			\node[point] (a) at (-1.5,0) [label = below:{$v$}]{};
			\node[point] (b) at (1.5,0) [label = below:{$w$}]{}; 
			\path[-latex] (a) edge [in = 135, out=45] node [above] {-1} (b);
			\path[-latex] (a) edge [in = -135, out=-45] node [below] {0} (b);
			\path[-latex] (a) edge[] node[below] {0} (b);
		\end{tikzpicture},
		\]
		and the $0$-residual graph of $-\bar\delta^{-1}_{v}(\Gamma,n)$ is 		
		\[\begin{tikzpicture}
			\node[point] (a) at (0,0) [label = below:{$w$}] {};
			\node[] (coe) at (-0.8,0) {\Large{$-\frac{0!}{0!0!}$}};
			\path[-latex] (a) edge[reflexive = 25, in = -45, out = 45] (a);
			\path[-latex] (a) edge[reflexive = 35, in = -55, out = 55] (a);
			\node[] (wei3) at (0.9,0) {\tiny{$0$}};
			\node[] (wei2) at (0.7,0) {\tiny{$0$}};
		\end{tikzpicture}
		\begin{tikzpicture}
			\node[] (coe2) at (-1,0) {\Large{$+\binom{2}{1}\frac{1!}{1!0!}$}};
			\node[point] (a) at (0,0) [label = below:{$w$}] {};
			\path[-latex] (a) edge[reflexive = 25, in = -45, out = 45] (a);
			\path[-latex] (a) edge[reflexive = 35, in = -55, out = 55] (a);
			\node[] (wei3) at (0.9,0) {\tiny{$0$}};
			\node[] (wei2) at (0.7,0) {\tiny{$0$}};
		\end{tikzpicture}
		\begin{tikzpicture}
			\node[point] (a) at (0,0) [label = below:{$w$}] {}; 
			\path[-latex] (a) edge[reflexive = 35, in = -55, out = 55]  (a);
			\node[] (coe3) at (-1,0) {\Large{$-\binom{2}{1} \frac{0!1!}{2!}$}};
			\node[] (wei3) at (0.9,0) {\tiny{$2$}};
		\end{tikzpicture}
		\begin{tikzpicture}
			\node[point] (a) at (0,0) [label = below:{$w$}] {}; 
			\path[-latex] (a) edge[reflexive = 35, in = -55, out = 55]  (a);
			\node[] (wei3) at (0.9,0) {\tiny{$2$}};
			\node[] (coe) at (-0.8,0) {\Large{$+\frac{1!1!}{3!}$}};
		\end{tikzpicture}
		\]
		
		The first term is obtained by collapsing the edge with decoration $-1$, with coefficient 
		\[(-1)^{-1}\frac{(-1+1)!}{u_1! \cdot u_2!},\]
		where $u_1+u_2 = (-1+2)-1 = 0$. So $u_1=u_2=0$.
		
		The second term is obtained by collapsing the edge with decoration $0$, each quotient graph has coefficient
		\[(-1)^0\frac{(0+1)!}{u_1!u_2!},\]
		where $u_1 + u_2 = (0+2)-1 = 1$. Since $W_k = 0$ for $2 \nmid k$, then we should assign $1$ to the edge with decoration $-1$, and assign $0$ to the edge with decoration $0$.
		
		The third term is obtained by collapsing two edges with decoration $-1$ and $1$ separately, with coefficient
		\[(-1)^{-1+0}\frac{(-1+1)!(0+1)!}{((-1+2)+(0+2)-1)!},\]
		and we assign $2$ to the remaining edge.
		
		And the last term is obtained by collapsing two edges with decoration $1$ and $1$ separately, with the coefficient 
		\[(-1)^{0+0} \frac{(0+1)!(0+1)!}{((0+2)+(0+2)-1)!},\]
		and we assign $3$ to the remaining edge.
		
		Therefore, by Example \ref{The value of Wk's}, one has
		\[W(\Gamma,n) = -W_0^2 + \frac{5}{6}W_2 = -\kuohao{\frac{\pi^2}{3}\widehat{E}_2}^2 + \frac{5}{6} \cdot 3! \cdot \frac{\pi^4}{45}E_4 = \pi^4\frac{-\widehat{E}_2^2+E_4}{9}.\]
	\end{eg}
	
	\begin{eg}
		\[\begin{tikzpicture}
			\node[point] (a) at (-1.5,0) [label = below:{$v$}]{};
			\node[point] (b) at (1.5,0) [label = below:{$w$}]{}; 
			\path[-latex] (a) edge [in = 135, out=45] node [above] {0} (b);
			\path[-latex] (a) edge [in = -135, out=-45] node [below] {0} (b);
			\path[-latex] (a) edge[] node [below] {0} (b);
			\node[] (c) at (2,0) [label = above:{$\bar\delta^{-1}_{v}$}] {$\longrightarrow$};
		\end{tikzpicture} 
		\begin{tikzpicture}
			\node[] (c) at (-2,0) {\LARGE{$-$}};
			\node[point] (a) at (-1.5,0) [label = below:{$v$}]{};
			\node[point] (b) at (1.5,0) [label = below:{$w$}]{}; 
			\path[-latex] (a) edge [in = 105, out=75] node [above] {-1} (b);
			\path[-latex] (a) edge [in = 135, out=45] node [below] {0} (b);
			\path[-latex] (a) edge [in = -135, out=-45] node [below] {0} (b);
			\path[-latex] (a) edge[] node[below] {0} (b);
		\end{tikzpicture}
		\]
		and the $0$-residual graph of $-\bar\delta^{-1}_{v}(\Gamma,n)$ is 
		\[
			\begin{tikzpicture}
				\node[point] (a) at (0,0) [label = below:{$w$}]{};
				\path[-latex] (a) edge[reflexive = 15, in=-45, out=45] (a);
				\node[] (decoration1) at (0.5,0) {\tiny{$0$}};
				\path[-latex] (a) edge[reflexive = 30, in=-50, out=50] (a);
				\node[] (decoration1) at (0.8,0) {\tiny{$0$}};
				\path[-latex] (a) edge[reflexive = 45, in=-55, out=55] (a);
				\node[] (decoration1) at (1.1,0) {\tiny{$0$}};
				\node[] (coe) at (-1.8,0) {$\kuohao{-\frac{0!}{0!0!0!}+\binom{3}{1}\frac{1!}{1!0!0!}}$};
			\end{tikzpicture}	
			\begin{tikzpicture}
				\node[point] (a) at (0,0) [label = below:{$w$}]{};
				\path[-latex] (a) edge[reflexive = 30, in=-50, out=50] (a);
				\node[] (decoration1) at (0.8,0) {\tiny{$0$}};
				\path[-latex] (a) edge[reflexive = 45, in=-55, out=55] (a);
				\node[] (decoration1) at (1.1,0) {\tiny{$2$}};
				\node[] (coe) at (-3,0) {$+\kuohao{-\binom{3}{1} \binom{2}{1}\frac{0!1!}{2!0!} + \binom{3}{2}\frac{1!1!}{3!0!}+ \binom{3}{2}\frac{1!1!}{1!2!}}$};
			\end{tikzpicture}
			\begin{tikzpicture}
				\node[point] (a) at (0,0) [label = below:{$w$}]{};
				\path[-latex] (a) edge[reflexive = 45, in=-55, out=55] (a);
				\node[] (decoration1) at (1.1,0) {\tiny{$4$}};
				\node[] (coe) at (-2,0) {$+\kuohao{-\binom{3}{2}\frac{0!1!1!}{4!}+ \frac{1!1!1!}{5!}}$};
			\end{tikzpicture}
		\]
		Therefore, 
		\[W(\Gamma,n) = -2W_0^3 + W_0W_2 + \frac{7}{60}W_4 = \pi^6 \frac{-10\widehat{E}_2^3 + 6\widehat{E}_2E_4 + 4E_6}{27 \cdot 5}\]
	\end{eg}
	
	\begin{eg}
		\[\begin{tikzpicture}
			\node[point] (a) at (-1.5,0) [label = below:{$v$}]{};
			\node[point] (b) at (1.5,0) [label = below:{$w$}]{}; 
			\path[-latex] (a) edge [in = 135, out=45] node [above] {0} (b);
			\path[-latex] (a) edge [in = -135, out=-45] node [below] {2} (b);
			\node[] (c) at (2,0) [label = above:{$\bar\delta^{-1}_{v}$}] {$\longrightarrow$};
		\end{tikzpicture} 
		\begin{tikzpicture}
			\node[] (c) at (-2,0) {\LARGE{$-$}};
			\node[point] (a) at (-1.5,0) [label = below:{$v$}]{};
			\node[point] (b) at (1.5,0) [label = below:{$w$}]{}; 
			\path[-latex] (a) edge [in = 135, out=45] node [above] {-1} (b);
			\path[-latex] (a) edge [in = -135, out=-45] node [below] {2} (b);
			\path[-latex] (a) edge[] node[below] {0} (b);
		\end{tikzpicture},
		\]
		and the $0$-residual graph of $-\bar\delta^{-1}_{v}(\Gamma,n)$ is 		
		\[\begin{tikzpicture}
			\node[point] (a) at (0,0) [label = below:{$w$}] {};
			\node[] (coe) at (-2.7,0) {\Large{$\left(-\frac{0!}{0!0!}+\frac{1!}{1!0!}+\frac{3!}{1!2!}+\frac{3!}{3!0!}\right)$}};
			\path[-latex] (a) edge[reflexive = 25, in = -45, out = 45] (a);
			\path[-latex] (a) edge[reflexive = 35, in = -55, out = 55] (a);
			\node[] (wei3) at (0.9,0) {\tiny{$2$}};
			\node[] (wei2) at (0.7,0) {\tiny{$0$}};
		\end{tikzpicture}
		\begin{tikzpicture}
			\node[point] (a) at (0,0) [label = below:{$w$}] {}; 
			\path[-latex] (a) edge[reflexive = 35, in = -55, out = 55]  (a);
			\node[] (coe3) at (-2.3,0) {\Large{$+\left(-\frac{0!1!}{2!}-\frac{1!3!}{4!}+\frac{1!3!}{5!}\right)$}};
			\node[] (wei3) at (0.9,0) {\tiny{$4$}};
		\end{tikzpicture}
		\]
		Therefore,
		\[W(\Gamma,n) = -4W_0W_2 + \frac{7}{10}W_4 = \pi^6 \frac{-8\widehat{E}_2E_4+8E_6}{45}.\]
	\end{eg}
	
	\begin{eg}
		\[\begin{tikzpicture}
			\node[point] (a) at (-1.5,0) [label = below:{$v$}]{};
			\node[point] (b) at (1.5,0) [label = below:{$w_1$}]{};
			\node[point] (c) at (0,2.6) [label = above:{$w_2$}]{};
			\path[-latex] (a) edge[] node[below] {0} (b);
			\path[-latex] (b) edge[] node[right] {0} (c);
			\path[-latex] (c) edge[] node[left] {0} (a);
			\node[] (arrow) at (2.25,1.3) [label = above:{$\op{\bar\delta_{v}^{-1}}$}] {$\longrightarrow$};
		\end{tikzpicture}
		\begin{tikzpicture}
			\node[] at (-2,1.3) {$-$};
			\node[point] (a) at (-1.5,0) [label = below:{$v$}]{};
			\node[point] (b) at (1.5,0) [label = below:{$w_1$}]{};
			\node[point] (c) at (0,2.6) [label = above:{$w_2$}]{};
			\path[-latex] (a) edge[] node[below] {0} (b);
			\path[-latex] (a) edge[in = 135, out = 45] node[above] {-1} (b);
			\path[-latex] (b) edge[] node[right] {0} (c);
			\path[-latex] (c) edge[] node[left] {0} (a);
		\end{tikzpicture}\]
		The residual graph of $-\bar\delta^{-1}_{v}(\Gamma,n)$ is 
		\[\begin{tikzpicture}
		\node[point] (a) at (-0.5,0) [label = below:{$w_2$}] {};
		\node[point] (b) at (0.5,0) [label = below:{$w_1$}] {};
		\path[-latex] (a) edge[out = 45, in = 135] node[above] {\tiny{$0$}} (b);
		\path[-latex] (a) edge[] (b);
		\node[] (decoration) at (0,-0.1) {\tiny{$0$}};
		\path[-latex] (a) edge[out = -45, in = -135] node[below] {\tiny{$0$}} (b);
		\node[] (coe1) at (-1,0) {\Large{$\frac{1!}{1!0!}$}};
		\node[] (coe0) at (1.2,0) {\Large{$+\frac{0!}{0!1!}$}};
		\node[point] (a0) at (2,0) [label = below:{$w_2$}] {};
		\node[point] (b0) at (3,0) [label = below:{$w_1$}] {};
		\path[-latex] (a0) edge[out = 45, in = 135] node[above] {\tiny{$0$}} (b0);
		\path[-latex] (a0) edge[] (b0);
		\node[] (decoration) at (2.5,-0.1) {\tiny{$-1$}};
		\path[-latex] (a0) edge[out = -45, in = -135] node[below] {\tiny{$1$}} (b0);
		\node[point] (a2) at (6.5,0) [label = below:{$w_2$}] {};
		\node[point] (b2) at (7.5,0) [label = below:{$w_1$}] {};
		\path[-latex] (a2) edge[out = 45, in = 135] node[above] {\tiny{$0$}} (b2);
		\path[-latex] (a2) edge[out = -45, in = -135] node[below] {\tiny{$0$}} (b2);
		\path[-latex] (b2) edge[reflexive = 15, out = 45, in = -45] (b2);
		\node[] (decoration3) at (8,0) {\tiny{$0$}};
		\node[] (coe2) at (4.8,0) {\Large{$+\kuohao{\frac{1!}{1!0!}-\frac{0!}{0!0!}}$}};
		\node[] (coe3) at (8.7,0) {\Large{$-\frac{0!1!}{2!}$}};
		\node[point] (a3) at (9.5,0) [label = below:{$w_2$}] {};
		\node[point] (b3) at (10.5,0) [label = below:{$w_1$}] {};
		\path[-latex] (a3) edge[out = 45, in = 135] (b3);
		\node[] (decoration31) at (10,0.4) {\tiny{$2$}};
		\path[-latex] (a3) edge[out = -45, in = -135] (b3);
		\node[] (decoration32) at (10,-0.4) {\tiny{$0$}};
	\end{tikzpicture}\]
	Therefore,
	\begin{align*}
		W(\Gamma,n) &= -\frac{1}{2}\left(-2W_0^3+W_0W_2+\frac{7}{60}W_4\right) + \frac{1}{2}\left(-4W_0W_2+\frac{7}{10}W_4\right)\\
		&=W_0^3-\frac{5}{2}W_0W_2+\frac{7}{24}W_4\\
		&= \pi^6\frac{\widehat{E}_2^3 - 3\widehat{E}_2E_4+2E_6}{3^3}.
	\end{align*}
	Moreover, we have
	\[\partial_{\Y}W(\Gamma,n) = \pi^4\frac{\widehat{E}_2^2-E_4}{3} = W(\delta(\Gamma,n)).\]
	\end{eg}
	
	\begin{eg}
		In general, consider the nonlinear system 
		\[\begin{tikzpicture}
			\node[point] (a) at (-1.5,0) [label = below:{$v$}]{};
			\node[point] (b) at (1.5,0) [label = below:{$w$}]{}; 
			\path[-latex] (a) edge [in = 135, out=45] node [above] {$n_1$} (b);
			\path[-latex] (a) edge [in = -135, out=-45] node [below] {$n_k$} (b);
			\node[] (vdot) at (0,0) {$\vdots$};
			\node[] (c) at (2,0) [label = above:{$\bar\delta^{-1}_{v}$}] {$\longrightarrow$};
		\end{tikzpicture} 
		\begin{tikzpicture}
			\node[] (c) at (-2,0) {\LARGE{$-$}};
			\node[point] (a) at (-1.5,0) [label = below:{$v$}]{};
			\node[point] (b) at (1.5,0) [label = below:{$w$}]{}; 
			\path[-latex] (a) edge [in = 105, out=75] node [above] {$-1$} (b);
			\path[-latex] (a) edge [in = 135, out=45] node [below] {$n_1$} (b);
			\node[] (vdot) at (0,0) {$\vdots$};
			\path[-latex] (a) edge [in = -135, out=-45] node [below] {$n_k$} (b);
		\end{tikzpicture}
		\]
		Let $E$ be the set of edges in $\bar\delta^{-1}_{v}(\Gamma,n)$, then
		\[W(\Gamma,n) = -\sum_{A \subset E \atop A \neq \emptyset , A \neq E} (-1)^{\omega(A)} \sum_{u \in P(E \setminus A, \omega(A)-1)} \frac{\prod_{e \in A}(n_e+1)!}{\prod_{e \in E \setminus A} u_e!} \prod_{e \in E \setminus A} W_{n_e+u_e}.\]
	\end{eg}